\documentclass[letterpaper, 10 pt, conference]{IEEEtran}
\IEEEoverridecommandlockouts  

\usepackage{cite}
\usepackage{graphics} 
\usepackage{epsfig} 
\usepackage{mathptmx} 
\usepackage{times} 
\usepackage{amssymb}  
\usepackage[cmex10]{amsmath}
\usepackage[ruled,vlined]{algorithm2e}
\usepackage{amsthm}
\usepackage{amsfonts}
\usepackage{latexsym}
\usepackage{mathrsfs}
\usepackage{float}
\usepackage{xcolor}
\usepackage{bm}
\usepackage[normalem]{ulem}
\usepackage{mathtools}
\usepackage{wrapfig}
\usepackage{xspace}
\usepackage{geometry}
\usepackage{graphicx}
\DeclareMathAlphabet{\mathcal}{OMS}{cmsy}{m}{n}

\ifCLASSOPTIONcompsoc
    \usepackage[caption=false, font=normalsize, labelfont=sf, textfont=sf]{subfig}
\else
\usepackage[caption=false, font=footnotesize]{subfig}
\fi
\usepackage{url}
\usepackage{mathtools}

\DeclarePairedDelimiter\floor{\lfloor}{\rfloor}

\theoremstyle{remark}
\newtheorem{rem}{Remark}
\theoremstyle{definition}
\newtheorem{defn}{Definition}

\newtheorem{ass}{Assumption}

\theoremstyle{plain}
\newtheorem{lem}{Lemma}
\newtheorem{thm}{Theorem}
\newtheorem{prop}{Proposition}
\newtheorem{cor}{Corollary}

\title{\LARGE \bf \vspace{18pt}
Control Barrier Functions for Sampled-Data Systems with \newline Input Delays        
}

\author{Andrew Singletary, Yuxiao Chen, and Aaron D. Ames
\thanks{Andrew Singletary, Yuxiao Chen, and Aaron D. Ames are with Department of Mechanical Engineering,
        California Institute of Technology, Pasadena CA 91125, U.S.A. Email addresses:
		{\tt \small \{asinglet,chenyx, ames\}@caltech.edu}}%
}

\begin{document}

\maketitle
\thispagestyle{empty}
\pagestyle{empty}

\begin{abstract}
This paper considers the general problem of transitioning theoretically safe controllers to hardware. Concretely, we explore the application of control barrier functions (CBFs) to sampled-data systems: systems that evolve continuously but whose control actions are computed in discrete time-steps. While this model formulation is less commonly used than its continuous counterpart, it more accurately models the reality of most control systems in practice, making the safety guarantees more impactful. In this context, we prove robust set invariance with respect to zero-order hold controllers as well as state uncertainty, without the need to explicitly compute any control invariant sets. It is then shown that this formulation can be exploited to address input delays in this system, with the result being CBF constraints that are affine in the input. The results are demonstrated in a high-fidelity simulation of an unstable Segway robotic system in real-time.
\end{abstract}

\section{INTRODUCTION}
\label{sec:introduction}
Control theory in practice is almost always implemented in the form of a digital controller on a physical system that evolves continuously. However, these systems are rarely treated as such due to the difficulties that arise in the formulation of controllers that act optimally for these types of systems. Generally, the system is controlled rapidly enough that the time discretization can be ignored, and the entire system can be treated as continuous, allowing for a much larger class of control techniques. This is especially true for robotic systems, where continuous controllers are often implemented at loop rates faster that 1 kHz. 

In the context of safety-critical control, however, it is important to model the system as accurately as possible, in order to extend the guarantees from theory to practice. Moreover, optimization-based controllers tend to be less robust than simple control laws such as PID, therefore call for a more accurate model. As an example, control barrier functions have emerged as a popular technique for guaranteeing safety of continuous-time \cite{ames2017cbf,wang2017safety}, discrete-time \cite{agrawal2017discrete,ahmadi2019safe}, and even sampled-data systems \cite{gurriet2019realizable}. This technique relies on the knowledge of a control invariant subset of the state-space, and can be formulated as a quadratic program wherein the CBF yields a constraint, affine in the input, that ensures the system stays in that set for all time. However, the computation of such sets is notoriously difficult for nonlinear systems \cite{fiacchini2010computation,mitchell2015summary}.

The method proposed in \cite{gurriet2018online} removes the need for a control invariant set, and instead relies on knowledge of a backup controller that takes the system to a small safe region of the state space. This method of implicitly defining a control barrier function has been successfully implemented in a variety of complex applications \cite{sing2019,drew_uav,taylor2020control}, but it suffers from a lack of robustness. 

Another reality of control theory in practice is the presence of input delay. Input delay for myopic, optimization-based controllers (such as Control Lyapunov Functions \cite{ogren2001control} and control barrier functions), is often handled by making the problem formulation robust to any value of the input delay in some bounded set. However, this treatment degrades performance due to conservatism and complicates the (already difficult) computation of the Lyapunov or barrier functions. In practice, the input delay for a system is relatively easy to identify, so it would be beneficial to formulate the problem with the knowledge of the input delay of the system. While many solutions to handling specific time delays have been proposed, in general they either require either linear systems \cite{jankovic2018control,fridman2010refined}, are applicable only to autonomous systems \cite{orosz2019safety}, require difficult construction \cite{jankovic2001control}, or rely on frequency-domain analysis that is not applicable to these optimization-based controllers \cite{zhong2006robust}.

The contributions of this paper are:
\begin{itemize}
    \item We propose a formulation of the implicitly defined control barrier function that is applicable to sampled-data systems, while retaining scalability.
    \item We use the concept of incremental stability \cite{angeli2002lyapunov} to prove robustness of the proposed backup controller-based CBF controller under state uncertainty.
    \item We are able to guarantee safety under a known input delay with much less conservatism under the proposed framework.
\end{itemize}

For the remainder of the paper, Section \ref{sec:zoh} details the conditions required for the control barrier function to provide safety guarantees for the sampled-data system that arises from a zero-order hold controller. The resulting constraint must be evaluated over a set of states, thus Section \ref{sec:state_uncertianty} provides a means for tractably evaluating this constraint. Section \ref{sec:input_delay} provides a method for handling known input delays, while maintaining the safety guarantee in the form of an affine constraint. The proposed method is then applied in Section \ref{sec:sim_results} to a high-fidelity Segway simulation, which demonstrates the result. Finally, Section \ref{sec:conclusion} provides a summary of the work and details future research directions.

\section{Sampled-Data Considerations for CBFs}
\label{sec:zoh}

\subsection{Background on Control Barrier Functions}
In this paper, we consider an affine dynamical system described by the following ordinary differential equation:
\begin{equation}\label{eq:dyn}
    \dot{x}=f(x)+g(x)u,\quad x\in\mathcal{X}\subseteq\mathbb{R}^n,~u\in\mathcal{U}\subseteq\mathbb{R}^m.
\end{equation}

Here, $f:\mathbb{R}^n\rightarrow \mathbb{R}^n$ and $g:\mathbb{R}^n\rightarrow \mathbb{R}^{n\times m}$ are locally Lipschitz functions. $\mathcal{U}$ is the set of admissible inputs, assumed to be a strict subset of $\mathbb{R}^m$. For a given dynamic system, given a feedback policy $u:\mathcal{X}\to\mathcal{U}$, $\phi_t^{u}(x_0)$ be the state evolution of this closed loop system with initial condition $x_0$ at time $t$, i.e., the flow map. With a slight abuse of terminology, we let $\phi_t^{u(\cdot)}$ denote the state flow under a fixed input signal $u(\cdot)$, as opposed to a fixed feedback policy.

Suppose that the system is required to stay within the set described by the function $\mathcal{S} = \{ x \in \mathbb{R}^n ~ | ~ h(x) \geq 0\}$ for all time. To guarantee this constraint is satisfied for all time, one way is to find a control invariant set $\mathcal{S}_I\subseteq \mathcal{S}$.
\begin{defn}
A set $\mathcal{S}_I$ is a control invariant set if there exists a control policy that keeps any trajectory starting within $\mathcal{S}_I$ inside $\mathcal{S}_I$ for all time. 
\end{defn}
If a control invariant subset described as the superlevel-set (see \cite{ames2017cbf}) of a function $h_I(x)$ is known, the CBF condition can be enforced to keep the state within $\mathcal{S}_I$:
\begin{equation}
\label{eq:barrier}
    \left.\frac{dh_I}{dx}\right\rvert_{x}\left( f(x) + g(x) u \right) + \alpha(h_I(x)) \geq 0,
\end{equation}
as an affine constraint on the input $u$. Here, $\alpha(\cdot)$ is an extended class-$\mathcal{K}$ function.

\begin{lem}
\label{lem:barrier}
If the input satisfies the CBF condition \eqref{eq:barrier} for all time, the system will stay in the set $h_I(x)$ for all time. Specifically, let $u(x)$ be a controller that satisfies \eqref{eq:barrier} and is applied to \eqref{eq:dyn} to yield a closed-loop dynamical system: $\dot{x} = f(x) + g(x) u(x) =: f_{\mathrm{cl}}(x)$ (assumed to be forward complete). Then if $x_0\in\mathcal{S}_I$, $\phi_t^{u}(x_0) \in \mathcal{S}_{I}$ for all $t \geq 0$. 
\end{lem}
See \cite{ames2017cbf} for proof.

As mentioned in the introduction, the computation of control invariant sets is very difficult in general. Instead, \cite{gurriet2018online} proposes utilizing a fixed backup controller $u_B(x)$ to enlarge a much smaller control invariant set $S_B= \{ x \in \mathbb{R}^n ~ | ~ h_B(x) \geq 0\}$ (which is much easier to compute) to the set of all states that can be safely brought to $S_B$ with the backup controller $u_B:\mathcal{X}\to\mathcal{U}$. This leads to a control invariant set defined implicitly using the system flow $\phi^{u_B}_t(x_0)$ under $u_B$ as
\begin{equation}
\begin{aligned}
    \label{eq:dsi_set}
    \mathcal{S}_I = \bigg\{ x \in \mathbb{R}^n ~ | ~ \Big( \forall \tau \in [0,T],\  & h\left(\phi^{u_B}_\tau\Big) \geq 0  \right) \textrm{ and }\\ 
    \Big( &h_B\left(\phi^{u_B}_T\right) \geq 0 \Big) \bigg\}
    \end{aligned}
\end{equation}
The CBF condition \eqref{eq:barrier} then is evaluated at a point $x_0$ as
\begin{equation}\label{eq:dsi}
\begin{aligned}
    \left. \frac{dh}{dx}\right\rvert_{\phi^{u_B}_\tau} \left. \frac{\partial \phi^{u_B}_t}{\partial x}\right\rvert_{x_0} &\left( f(x_0)+g(x_0)u \right) + \alpha(h(\phi^{u_B}_\tau(x_0))) \geq 0 \\
    \left. \frac{dh_B}{dx}\right\rvert_{\phi^{u_B}_T} \left. \frac{\partial\phi^{u_B}_T}{\partial x}\right\rvert_{x_0}&\left( f(x_0)+g(x_0)u \right) + \alpha(h_B(\phi^{u_B}_T(x_0))) \geq 0 
    \end{aligned}
\end{equation}
Any input that satisfies \eqref{eq:dsi} for all time $\tau \in [0,T]$ will keep the system in the set $\mathcal{S}_I$. If this constraint is satisfied for all time, the state is kept in the set $\mathcal{S}_I$ indefinitely. See \cite[p. 6]{gurriet2018online}.

\subsection{CBFs Along a Zero-Order Hold Backup Controller}

While condition \eqref{eq:dsi} can be used to guarantee safety for continuous-time systems, it relies on the input being computed and applied continuously. In the presence of a zero-order hold controller, the description of the invariant set need to be modified. The zero-order hold backup controller, denoted as $\overline{u_B}(x(\cdot),t)$, simply takes on the value of $u_B(x)$ every $\Delta_t$ seconds, and holds that value until its next update:
\begin{align}
 \overline{u_B}(x(\cdot),t)=u_B(x(\floor*{t/\Delta_t}\Delta_t)),
\end{align}
where $\floor*{\cdot}$ is the largest integer not greater than the argument.

\begin{rem}
The flow of the system under the zero-order backup controller $\phi^{\overline{u_B}}_t$ is well-defined, as unique solutions to sampled-data systems exist so long as the underlying controller is piecewise continuous (c.f. \cite[p. 16]{lars2011nonlinear})
\end{rem}

The updated control invariant set under the zero-order hold backup controller can be written as
\begin{equation}\label{eq:SIzoh}
\begin{aligned}
    \overline{\mathcal{S}_I} = \bigg\{ x \in \mathbb{R}^n ~ | ~ \Big( \forall \tau \in [0,T],\  & h\left(\phi^{\overline{u_B}}_\tau\Big) \geq 0  \right) \textrm{ and }\\ 
    \Big( &h_B\left(\phi^{\overline{u_B}}_T\right) \geq 0 \Big) \bigg\}
\end{aligned}
\end{equation}
This is a control invariant set for the system under a zero-order controller with the same sampling time $\Delta_t$ as $\overline{u_B}$. The proof follows from \cite[Theorem 1]{gurriet2018online}.

\begin{rem}
\label{rem:nonsmooth}
While the flow of the system under the zero-order hold backup controller is Lipschitz \cite{abbaszadeh2016lipschitz}, it is nonsmooth. Because of this, the barrier function itself is nonsmooth, and thus $\dot{h}$ cannot be expressed at finitely many points, which correspond to when the controller is updated. Despite this, a nonsmooth barrier function is valid if $\dot{h} \geq -\alpha(h)$ almost everywhere. For proof, see \cite[Lemma 2.2]{glotfelter2017nonsmooth}.
\end{rem}

To enforce the CBF condition, $\dot{h}$ needs to be computed, which requires $\frac{\partial \phi^{\overline{u_B}}_t}{\partial x}$ to be evaluated. This expression is continuous over each controller sampling time, and can be computed using finite-differences \cite{leveque1998finite}.
To do this, simply integrate forward $n+1$ initial conditions under $\overline{u_B}$  to evaluate $\frac{\partial \phi_{(i+1)\Delta_t}^{u_B}(x)}{\partial \phi_{i\Delta_t}^{u_B}(x)}$ at each time-step. Then, use the chain rule to get 
\begin{align}
    \frac{\partial \phi_{i\Delta_t}^{u_B}(x)}{\partial x} = \prod_{n=0}^{i-1}\frac{\partial \phi_{(n+1)\Delta_t}^{u_B}(x)}{\partial \phi_{n\Delta_t}^{u_B}(x)}
\end{align}

\subsection{Enforcing the Barrier Condition with Zero-order Hold}

The last caveat to consider in the implementation of the CBF condition is the fact that the condition must be met over the entire time horizon of the zero-order hold controller. Note that this consideration must be taken regardless of the method used for expressing the robust control barrier function, and was a subject of prior research of the authors \cite{gurriet2019realizable}.

Consider verifying the barrier function over the horizon of a single time-step of the zero-order hold controller with sample time $\Delta_t$, 
\begin{equation}
h(\phi_{\tau}^{\overline{u_B}}(x_0)) \geq 0 \qquad  \forall ~ \tau \in [0,\Delta_t].
 \label{eq:single_step}
\end{equation}

The robust satisfaction of the above condition can be verified by checking the stronger condition shown in \cite{gurriet2019realizable},
\begin{equation}
 h(\mathcal{R}(x_0,\Delta_t)) \geq 0,
 \label{eq:reach_cond}
\end{equation}
where $\mathcal{R}(x_0,\Delta_t)$ is the set of states reachable from $x_0$ in time $\Delta_t$ with any input $u \in \mathcal{U}$. 

\begin{rem}
\label{rem:conserve}
This condition adds conservatism to the barrier formulation. While checking only the points $\phi^u_{\tau}(x_0) \ \forall \tau \in [0, \Delta_t]$ would result in a more performant condition, this would make the constraint no longer affine, due to its dependence on the decision variable $u$.
\end{rem}

Let $\overline{u_B}(\cdot)$ be the input signal w.r.t. the nominal state flow, the robust CBF condition defined from the set $\overline{\mathcal{S}_I}$, 
\begin{equation}\label{eq:robust_CBF}
\resizebox{1\columnwidth}{!}{$
\begin{aligned}
    \left. \frac{dh}{dx}\right\rvert_{\phi^{\overline{u_B}(\cdot)}_\tau(\mathbf{x_0})} \frac{\partial \phi^{\overline{u_B}(\cdot)}_\tau(\mathbf{x_0})}{\partial x} &\left( f(\mathbf{x_0})+g(\mathbf{x_0})u \right) + \alpha(h(\phi^{\overline{u_B}(\cdot)}_\tau(\mathbf{x_0}))) \geq 0 \\
    \left. \frac{dh_B}{dx}\right\rvert_{\phi^{\overline{u_B}(\cdot)}_T(\mathbf{x_0})} \frac{\partial \phi^{\overline{u_B}(\cdot)}_T(\mathbf{x_0})}{\partial x}&\left( f(\mathbf{x_0})+g(\mathbf{x_0})u \right) + \alpha(h_B(\phi^{\overline{u_B}(\cdot)}_T(\mathbf{x_0}))) \geq 0
\end{aligned}
$}
\end{equation}
where $\mathbf{x_0} = \mathcal{R}(x_0,\Delta_t)$, and the first inequality must hold for all $\tau \in [0,T]$.

\begin{prop}
Let $(\overline{u}_i)_{i = 0}^\infty$ be a sequence of inputs that satisfies \eqref{eq:robust_CBF} at the beginning of each time-step and is applied with zero-order hold to the system \eqref{eq:dyn}. If $x_0 \in \overline{\mathcal{S}_I}$, then $\phi_t^{\overline{u}}(x_0) \in \overline{\mathcal{S}_I}$ for all $t \geq 0$.
\end{prop}
\begin{proof}
Consider any time interval of a single time-step $\mathcal{T} = [t_0,t_0+\Delta_t]$, and assume $x_{t_0} \in \overline{\mathcal{S}_I}$.

Denote $\Phi:= \{\bigcup_{t \in \mathcal{T}}\phi_t^{\overline{u}_i}(x_{t_0}) \}$. Since $\mathbf{x_{t_0}} = x_{t_0}+\mathcal{R}(x_{t_0},\Delta_t)$ covers all states reachable from time $t_0$, $\Phi \subset \phi_t^{\overline{u}_i}(\mathbf{x_{t_0}})$,
Therefore, if $\overline{u}_i$ satisfies \eqref{eq:robust_CBF} for $\mathbf{x_{t_0}}$, then the CBF condition holds for $\Phi$ as well. Thus, by the invariance of $\overline{\mathcal{S}_I}$ and Lemma \ref{lem:barrier}, $\phi_t^{\overline{u}_i}(x_{t_0}) \in \overline{\mathcal{S}_I}$ for all $t \in \mathcal{T}$.

Since this condition is met over the entire sequence of time-steps, $\phi_t^{\overline{u}}(x_0) \in \overline{\mathcal{S}_I}$ for all $t \geq 0$. 
\end{proof}

\begin{rem}
It is possible that, for some $x_0 \in \mathcal{S}_I$, that $x_0 \in \mathcal{S}_I$ but $\mathbf{x_0} \notin \mathcal{S}_I$. In this case, the system is inside of its control invariant set, but the CBF condition \eqref{eq:robust_CBF} cannot be satisfied. This is due to conservatism mentioned in Remark \ref{rem:conserve}. However, when this occurs, the backup control action can be taken. Thus, the system will stay safe for all time. Furthermore, this occurs on a very small set at the boundary of $\mathcal{S}_I$, which the strengthening term $\alpha(\cdot)$ makes difficult to reach.
\end{rem}

Note that the condition is evaluated here over a set, rather than a single point. The evaluation of $\phi^{\overline{u_B}(\cdot)}_\tau(\mathbf{x_0})$ poses the most difficulty, as it involves robustly integrating over a set. This makes techniques like interval arithmetic \cite{jaulin2001interval,gurriet2018towards} difficult to implement, due to numerical issues.


\section{State Uncertainty}
\label{sec:state_uncertianty}
In this  section, the concept of incremental stability will be used to show that for any $\tau \in [0,T]$, a fixed control signal $u(\cdot)$ and an uncertainty set $\bm{\Delta}_x$, $\phi^{u(\cdot)}_\tau(x_0+\bm{\Delta}_x) \subseteq \phi^{u(\cdot)}_\tau(x_0)+\bm{\Delta}_x$. Thus, the CBF condition \eqref{eq:robust_CBF} can be evaluated over the entire set as an affine condition without requiring robust integration over sets.



\subsection{Incremental stability with Lyapunov functions}
To show that
safety can be guaranteed for a small neighborhood of initial conditions, we adopt the concept of incremental stability (c.f. \cite{angeli2002lyapunov}). 

\begin{defn}
Given the dynamic system in \eqref{eq:dyn}, the system is incrementally stable inside a set $\mathcal{X}\subseteq \mathbb{R}^n$ if $\forall T\ge 0$, $\forall~ x_1,x_2\in\mathcal{X}$ and $u(\cdot):[0,T]\to\mathbb{R}^m$ such that $\phi_t^{u(\cdot)}(x_1)$ and $\phi_t^{u(\cdot)}(x_2)$ stay inside $\mathcal{X}$, the evolution of the state satisfies $\left\|\phi^{u(\cdot)}_t(x_1)-\phi^{u(\cdot)}_t(x_2)\right\|\le \beta(\left\|x_1-x_2\right\|,t)$, where $\beta:\mathbb{R}\times[0,T]\to\mathbb{R}$ is nonincreasing in $t$ and $\forall~t\in[0,T], ~ \beta(\cdot,t)$ is a class-$\mathcal{K}$ function. 
\end{defn}

\begin{prop}
Suppose there exists a Lyapunov function $V:\mathcal{X}\to\mathbb{R}$ that satisfies $c_1||x||\le V(x)\le c_2||x||$ for some $c_2\ge c_1> 0$. For two initial conditions $x_1,x_2\in\mathcal{X}$ and an input signal $u(\cdot)$ such that $\phi^{u(\cdot)}(x_1),\phi^{u(\cdot)}(x_2)$, and $\phi^{u(\cdot)}(x_1)-\phi^{u(\cdot)}(x_2)\in\mathcal{X}$, let $V(t)=V(\phi^{u(\cdot)}_t(x_1)-\phi^{u(\cdot)}_t(x_2))$. If $\dot{V}\le 0$, then the system is locally incrementally stable in $\mathcal{X}$.
\end{prop}
\begin{proof}
The proof follows from the fact that $V(\cdot)$ and $||\cdot||$ are equivalent norms.
\end{proof}

In the context of control barrier functions with a backup strategy, if the system is incrementally stable, then given a nominal initial condition $x_0$ and an uncertainty set characterized as a level-set of the Lyapunov function, $\forall x\in\{x|V(x-x_0)\le \epsilon\}$, for any input signal $u(\cdot)$, $\phi^{u(\cdot)}_t(x)\in\{x|V(x-\phi^{u(\cdot)}_t(x_0))\le\epsilon\}$. We shall show in Section \ref{sec:generalizing_uncertainty} how this result can simplify the robust CBF condition in \eqref{eq:robust_CBF}, which requires that the CBF condition hold for a small set $\mathcal{R}(x_0,\Delta_t)$ around the nominal initial condition. 

\subsection{Gaining incremental stability via pre-feedback}\label{incremental}
The Segway model that considered in Section \ref{sec:sim_results} is not incrementally stable, but pre-feedback can be used to make it so. Since the error dynamics are being considered, the nonlinear dynamics are linearized to simplify the analysis. Given a set $\mathcal{X}\subseteq\mathbb{R}^n$ of states, multiple linear dynamics models $\dot{x}=A_i x+B_i u,~i=1,...,N$ can be obtained by considering the extreme points of $\mathcal{X}$. Given a quadratic Lyapunov function $V=x^\intercal P x$ where $P$ is symmetric and positive definite, and an input set hyperbox defined as $\mathcal{U}=\{-u^{\max}\le u\le u^{\max}\}\subseteq \mathbb{R}^m$, we develop the following Linear Matrix Inequality (LMI) to search for a pre-feedback gain that guarantees incremental stability for the system:
\begin{equation}\label{eq:LMI}
\begin{aligned}
    \mathop {\min }\limits_{K\in\mathbb{R}^{n\times m}} ~&{||\Lambda P^{-\frac{1}{2}}K^\intercal||_{\infty}}\\
    \mathrm{s.t.}~&\forall i=1,...,N, P(A_i+B_i K)+(A_i+B_i K)^\intercal P\le 0,
    \end{aligned}
\end{equation}
where $\Lambda = 
\textrm{diag}(\frac{1}{u_1^{max}},~ ... ~, ~\frac{1}{u_m^{max}})$.

The cost function is chosen due to the fact that
\begin{equation}\label{eq:K_opt}
   \{\mathop {\max }\limits_x ~ |K_i x|~ \mathrm{s.t.}~ x^\intercal P x\le 1\}=\sqrt{K_iP^{-1}K_i^\intercal},
\end{equation}
which means that the pre-feedback is available within the level set $\{x|x^\intercal P x\le \min\limits_{i=1,...,N}{\frac{u^{\max}_i}{\sqrt{K_i P^{-1}K_i^\intercal}}\}}$. Therefore, minimizing the cost function in \eqref{eq:LMI} is maximizing the size of the level-set of the Lyapunov function in which the pre-feedback is available.

\begin{prop}
Given a dynamic system as described in \eqref{eq:dyn} with $\mathcal{U}=\{-u^{\max}\le u\le u^{\max}\}$, a set $\mathcal{X}\subseteq \mathbb{R}^n$, and a Lyapunov function $V(x)=x^\intercal P x$, $P\ge 0$, assume that $\forall~x_1,x_2\in \mathcal{X}$, $\forall~ u\in\mathcal{U}$, $f(x_1)+g(x_1)u-f(x_2)-g(x_2)u\in Conv{(A_i)(x_1-x_2)}$. Then, with a $K$ solved with \eqref{eq:K_opt}, the system with pre-feedback $\dot{x}=f(x)+g(x)(u+Kx)$ is incrementally stable within $\mathcal{X}\cap\{x|x^\intercal P x\le \min\limits_{i=1,...,N}{\frac{u^{\max}_i}{\sqrt{K_i P^{-1}K_i^\intercal}}\}}$.
\end{prop}

\begin{proof}
Since the $A$ and $B$ matrix enters linearly into the Lyapunov condition in \eqref{eq:LMI}, by the assumption that $f(x_1)+g(x_1)u-f(x_2)-g(x_2)u\in Conv{(A_i)(x_1-x_2)}$, convexity shows that $\dot{V}(x_1-x_2)\le 0$, which shows incremental stability.
\end{proof}

\subsection{Generalizing State Uncertainty}\label{sec:generalizing_uncertainty}

The CBF condition shown in Equation \eqref{eq:robust_CBF} is shown for a specific uncertainty set $\mathbf{x_0} = R(x_0,\Delta_t)$ that arises from the sampled-data nature of the system. For an incrementally stable dynamic system, the CBF condition is rewritten as
\begin{equation}\label{eq:robust_general_cbf}
\resizebox{1\columnwidth}{!}{$
\begin{aligned}
    \left. \frac{dh}{dx}\right\rvert_{\left.\phi^{\overline{u_B}(\cdot)}_\tau(x_0)\right\rvert_{\mathbf{x_0}}} \left. \frac{\partial \phi^{\overline{u_B}(\cdot)}_\tau}{\partial x}\right\rvert_{\mathbf{x_0}}&\left( f(\mathbf{x_0})+g(\mathbf{x_0})u \right) + \alpha(h(\left.\phi^{\overline{u_B}(\cdot)}_\tau(x)\right\rvert_{\mathbf{x_0}})) \\
    \left. \frac{dh_B}{dx}\right\rvert_{\left.\phi^{\overline{u_B}(\cdot)}_T(x_0)\right\rvert_{\mathbf{x_0}}}\left. \frac{\partial \phi^{\overline{u_B}(\cdot)}_\tau}{\partial x}\right\rvert_{\mathbf{x_0}}&\left( f(\mathbf{x_0})+g(\mathbf{x_0})u \right) + \alpha(h_B(\left.\phi^{\overline{u_B}(\cdot)}_T(x)\right\rvert_{\mathbf{x_0}}))
\end{aligned}
$}
\end{equation}
The new set in which the constraint is being evaluated is $\mathbf{x_0} := \mathcal{R}(x_0,\Delta_t) + \bm{\Delta}_x$, where $\bm{\Delta}_x \subset \mathbb{R}^n$ is the state uncertainty set such that the estimated value of the state $\tilde{x} \in x+\bm{\Delta}_x$, with $x$ being the true state. The other major difference from Equation \eqref{eq:robust_CBF} is that the flow over the backup trajectory is now being computed for the nominal value of $x_0$, and it is simply being evaluated over the set $\phi^u_t(x_0)+\bm{\Delta}_x$. This greatly simplifies the computation, and makes the constraint tractable in real-time.

\begin{thm}\label{thm:robust_invariance}
Let $\overline{u}(\cdot)$ be a input signal with zero-order hold that satisfies \eqref{eq:robust_general_cbf}. If the system is incrementally stable in $\overline{\mathcal{S}_I}$, then $\phi_t^{\overline{u}(\cdot)}(x_0) \in \overline{\mathcal{S}_I}$ for all $t \geq 0$.
\end{thm}
\begin{proof}
From incremental stability, we have $\forall ~ x_1,x_2 \in \overline{\mathcal{S}_I}$, and for $\beta:\mathbb{R}\times[0,T]\to\mathbb{R}$ nonincreasing in $t$, and $\forall ~ t_1,t_2 \in \mathbb{R}_+$ with $t_2>t_1$,
\begin{alignat}{2}
    &&\left\|\phi^{u(\cdot)}_t(x_1)-\phi^{u(\cdot)}_t(x_2)\right\|&\le \beta(\left\|x_1-x_2\right\|,t) \notag \\
    \ArrowBetweenLines[\Downarrow]
    && \left\|\phi^{u(\cdot)}_{t_2}(x_1)-\phi^{u(\cdot)}_{t_2}(x_2)\right\|&\le \left\|\phi^{u(\cdot)}_{t_1}(x_1)-\phi^{u(\cdot)}_{t_1}(x_2)\right\|
\end{alignat}
Therefore, 
\begin{alignat}{2}
     &\left\|\phi^{u(\cdot)}_{0}(x_1)-\phi^{u(\cdot)}_{0}(x_2)\right\|= \left\|x_1-x_2\right\| \notag \\
    \ArrowBetweenLines[\Downarrow]
    & \phi^{\overline{u_B}(\cdot)}_t(\mathbf{x_0}) \subset \left.\phi^{\overline{u_B}(\cdot)}_t(x_0)\right\rvert_{\mathbf{x_0}}
\end{alignat}

Fix any $\mathbf{x_0}$, $\bar{u}$, $t$. For brevity, let $\Phi_1 := \phi^{\overline{u_B}(\cdot)}_t(\mathbf{x_0})$ and $\Phi_2 := \left.\phi^{\overline{u_B}(\cdot)}_t(x_0)\right\rvert_{\mathbf{x_0}}$, and let $\mathbf{\dot{x}} := \left. \frac{\partial \phi^{\overline{u_B}(\cdot)}_t}{\partial x}\right\rvert_{\mathbf{x_0}} \left(f(\mathbf{x_0})+g(\mathbf{x_0})u\right)$.

\begin{equation*}
  \Phi_1 \subset \Phi_2\;\;\Rightarrow\;\;
  \left.\frac{dh}{dx}\right\rvert_{\Phi_1} \subset \left.\frac{dh}{dx}\right\rvert_{\Phi_2}\;\;\Rightarrow\;\;
  \left.\frac{dh}{dx}\right\rvert_{\Phi_1} \mathbf{\dot{x}}\subset \left.\frac{dh}{dx}\right\rvert_{\Phi_2} \mathbf{\dot{x}}.
\end{equation*}

Following the same logic, we have 
\begin{equation*}
  \Phi_1 \subset \Phi_2\quad
    \Rightarrow\quad
  \alpha(h(\Phi_1)) \subset \alpha(h(\Phi_2))
\end{equation*}
Thus, 
$$
  \left.\frac{dh}{dx}\right\rvert_{\Phi_1} \mathbf{\dot{x}}+\alpha(h(\Phi_1)) \subset \left.\frac{dh}{dx}\right\rvert_{\Phi_2}\mathbf{\dot{x}} + \alpha(h(\Phi_2)).
$$

Therefore, any $\mathbf{x_0},\bar{u}(\cdot)$ that meets condition \eqref{eq:robust_general_cbf} will also meet condition $\eqref{eq:robust_CBF}$, and by Proposition 1, $\phi_t^{\overline{u}(\cdot)}(x_0) \in \overline{\mathcal{S}_I}$ for all $t \geq 0$
\end{proof}

\section{Input Delay}
\label{sec:input_delay}



In the previous sections, it is shown that safety can be guaranteed for sampled-data systems with an affine constraint using control barrier functions. This section will extend the safety guarantees to systems with known time-delay, without simply making the barrier robust to a set of possible input delays, which would degrade system performance. Note that the analysis in this section can be done for any robust control invariant set, not limited to the version discussed in previous sections.

\subsection{Preliminaries}

\begin{ass}
Suppose that the system has a time delay equal to some integer $n$ multiple of the controller period $\Delta_t$. Therefore, the system evolves with dynamics
\begin{equation}
\label{eq:zohdelaydyn}
    \dot{x} = f(x) + g(x)\bar{u}(x,t-n\Delta_t)
\end{equation}
for zero-order hold controller $\bar{u}$.
\label{ass:integer}
\end{ass}
This is a reasonable assumption, especially for the time delay caused by the numerical computation of the digital controller. Moreover, rounding of the time-delay can always be made robust with an addition to the state uncertainty.

Since it is not possible to provide any input to the system before time $t = n \Delta_t$, one more assumption is required.

\begin{ass}
From any initial set of states $\mathbf{x_0} = x_0 + \bm{\Delta}_x \subset \overline{\mathcal{S}_I}$ , we require
\begin{equation}
    \phi^0_{n\Delta_t}(\mathbf{x_0}) \subset \overline{\mathcal{S}_I}.
\end{equation}
Here, the $0$ in $\phi^0_{n\Delta_t}$ denotes the fact that a control input of zero is applied to the system during this time.
\label{ass:passive}
\end{ass}

In practice, this is not a restrictive assumption since the initial condition can be set well within the safe set. Moreover, if this is not met, there is no hope to keep the system safe whatsoever.

In order to obtain the state at which the control input will be applied, the most recent $n \Delta_t$ inputs must be stored in a vector $\bar{u}_H$. Since no input can be applied during time $t \in [0,n \Delta_t]$, the input vector is initialized to all zeros. With this, the state at which the $i$th computed control action will be applied can be expressed as
\begin{align}
    x_{(i+n) \Delta_t} = \phi^{\bar{u}_H}_{{n \Delta_t}}(x_{i\Delta_t})
    \label{eq:int}
\end{align}
The input history vector $\bar{u}_H$ is executed under zero-order hold, just as the inputs are applied to the system. Starting with the initial state, the algorithm for handling input delay is now described.

\subsection{Algorithm Overview}

At initial time-step $t_0$, the control action to be implemented at time $t = n \Delta_t$ is computed. To keep the system safe, the barrier conditions must be evaluated at state $x_{n \Delta_t}$, which is computed using Equation \eqref{eq:int}. Note that the constraint itself does not need to be altered, and is still affine. The only extra step is the integration from $x_0$ to $x_{n \Delta_t}$.

The input chosen by the quadratic program at time $t_0$ is then placed at the head of the $\bar{u}_H$ buffer, after each previous value is shifted backwards. Thus, the oldest value in the input buffer is lost, as it has already taken effect on the system.

The computation for all future time-steps is outlined in Algorithm \ref{ag:delay}.

\begin{algorithm}
\SetAlgoLined
 double $u_H[n]$ = \{0\}\; $i\gets 0$\; 
 \While{true}{
  $x_{(i+n)\Delta_t} = \phi^{\bar{u}_H}_{{n \Delta_t}}(x_{i\Delta_t})$\;
  Compute safe action with $u_{des},x_{(i+n)\Delta_t}$ using \eqref{eq:robust_general_cbf}\; 
  update $u_H$\;
  $i=i+1$\;
 }
 \caption{CBF with Input Delay of $n\Delta_t$}
 \label{ag:delay}
\end{algorithm}

While this algorithm may seem trivial, it is only made possible by treating the system as a sampled-data system. The continuous case of this algorithm would be much more complex, as there is no finite time-history of inputs to integrate over. Safety under this algorithm is summarized with the following theorem.

\begin{thm}
Given a control invariant set $\overline{\mathcal{S}_I}$, if inputs are chosen with Algorithm \ref{ag:delay}, and the system model is accurate, then the system \eqref{eq:zohdelaydyn} remains safe, i.e. $\phi^{\overline{u}}_{t}(x_0) \in \overline{\mathcal{S}_I}$ for all $t \geq 0$.
\end{thm}

\begin{proof}
Assume by contradiction that for some time $t = m \Delta_t$, $x_{m\Delta_t} \notin \overline{\mathcal{S}_I}$. Let this be the first time in which $x \notin \overline{\mathcal{S}_I}$, thus $x_{k\Delta_t} \in \overline{\mathcal{S}_I} ~ \forall k < m$.

Because system integration is accurate, 
\begin{align*}
    x_{m\Delta_t} &= \phi^{\overline{u}^m_H}_{n\Delta_t}(x_{(m-n)\Delta_t})  \quad \overline{u}^m_H = [u_{m-2n+1},...,u_{m-n}] \\
    x_{(m-1)\Delta_t} &= \phi^{\overline{u}^{m-1}_H}_{n\Delta_t}(x_{(m-n-1)\Delta_t})  \quad \overline{u}^{m-1}_H= [u_{m-2n},...,u_{m-n-1}]
\end{align*}  
Since $x_{(m-1)\Delta_t} \in \overline{\mathcal{S}_I}$, the control inputs chosen up until time $t_{m-n-1}$ keep the system in $\overline{\mathcal{S}_I}$. Therefore, the control action $u_{m-n}$ must cause the system to exit $\overline{\mathcal{S}_I}$. The CBF condition at $t=(m-n-1)\Delta_t$ is based on $\phi^{\overline{u}^{m-1}_H}_{n\Delta_t}(x_{(m-n-1)\Delta_t})$, the estimated $x_{(m-1)\Delta_t}$, but since the model is assumed to be correct, it is equal to the actual state. However, if $u_{m-n}$ was computed via CBF condition, then $\phi^{\overline{u}^m_H}_{n\Delta_t}(x_{(m-n)\Delta_t}) \in \overline{\mathcal{S}_I}$ by Theorem \ref{thm:robust_invariance}. This implies that $x_{m\Delta_t} \neq \phi^{\overline{u}^m_H}_{n\Delta_t}(x_{(m-n)\Delta_t})$, which contradicts the assumption that the model is accurate.

\end{proof}

\begin{figure*}[t]
     \centering
  \includegraphics[width=1\textwidth]{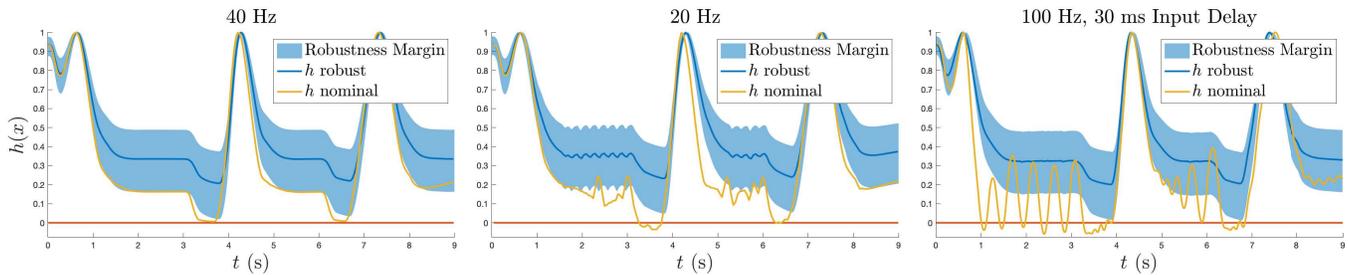}
  \caption{Results from simulations with three different controller frequencies, and one with input delay.}
  \label{fig:onlyplot}
\end{figure*}

The set $\overline{\mathcal{S}_I}$ described in Section \ref{sec:zoh} is an example of one such control invariant set robust to zero-order hold, but this theorem holds for any other such set.

It is important to recognize the fact that, under this algorithm, one is effectively performing open-loop control over the time-horizon of the input delay. However, due to the state uncertainty result from Section \ref{sec:state_uncertianty}, it is possible to guarantee safety for a range of possible initial conditions that the system is expected to lay within at the time of the control input being enacted. Thus, we have the following extension:

\begin{cor}
Given an invariant set $\overline{\mathcal{S}_I}$, robust to state uncertainty  $\bm{\Delta}_x$, if inputs are chosen with Algorithm \ref{ag:delay}, and $\phi^{\bar{u}_H}_{{n \Delta_t}}(x_0) \in x_{n\Delta t} + \bm{\Delta}_x$ for any $x_0 \in \overline{\mathcal{S}_I}$ (i.e. the system integration is accurate up to the set uncertainty set $\bm{\Delta}_x$), then the system \eqref{eq:zohdelaydyn} remains safe, i.e. $x(t) \in \overline{\mathcal{S}_I}$ for all $t \geq 0$.
\end{cor}
\begin{proof}
The proof follows directly from the Theorem 1, and simply utilizes the guarantees over state uncertainty from Section \ref{sec:state_uncertianty}, or from \cite{gurriet2019realizable} for more general control barrier functions.
\end{proof}

\section{Simulation of Results on Segway Robot}
\label{sec:sim_results}


\begin{wrapfigure}{L}{0.21\textwidth}
     \centering
  \includegraphics[width=.21\textwidth]{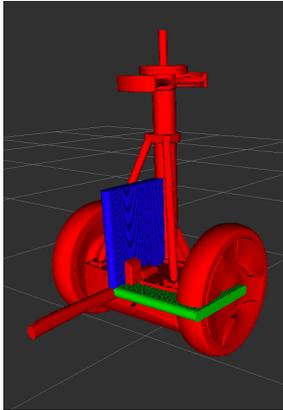}
  \caption{Segway in simulation. The sequence of axes show the system along the backup trajectory.}
  \label{fig:notonlyplot}
\end{wrapfigure}

The simulation is done in a ROS-based simulation environment. The full, nonlinear dynamics are integrated under zero-order hold at a variable sampling time $\Delta_t$. The true state of the system is not known to the controller, only the state estimate from an extended Kalman filter. This state observer receives noisy sensor data based on the true state of the system. A pre-feedback gain was computed following Section \ref{incremental}.

The Segway has 4 states $x=[p,\dot{p},\theta,\dot{\theta}]^\intercal$. The safe set is described by $\mathcal{S} =\{x| 1-4p^2\geq 0\}$, which enforces the robot position $p$ to stay within a 0.5 m range from the origin. The robust barrier condition \eqref{eq:robust_general_cbf} is evaluated over sets using the interval arithmetic library libaffa \cite{gay2006libaffa}. The constraint is imposed at the 10 closest points to the boundary of the safe set along the backup trajectory.

Figure \ref{fig:onlyplot} shows the result of three simulations with the nominal CBF conditions \eqref{eq:dsi}, and the robust condition \eqref{eq:robust_general_cbf}. The robust condition is set to handle a state uncertainty set based on the uncertainty caused by the zero-order hold and the Kalman filter. At 40 Hz, the Segway is able to stay within the set with the nominal controller, but it is unable to maintain invariance at 20 Hz, or in the presence of an input delay of 30 ms. The robust barrier is able to maintain safety for not just the nominal trajectory, but over the entire robustness margin.

\section{Conclusion}
\label{sec:conclusion}
In this paper, the authors first introduced a robust variation of the backup-controller-based control barrier function that guarantees safety in the presence of a zero-order controller, as well as state uncertainty. The barrier can be used without computing any robust control invariant sets, and can be evaluated in real-time as an affine constraint. An extension of control barrier functions for systems with known input delay was then introduced, made possible only by treating the system as sampled-data as opposed to continuous. Finally, the theory was validated in a real-time high-fidelity simulation environment of a Segway robot. For future work, the authors would like to explore solving this problem using discrete-time control barrier functions in a nonlinear program. 

\bibliography{refs.bib}
\bibliographystyle{IEEEtran}

\end{document}